\documentclass[journal]{IEEEtran} 

\usepackage{comment}
\usepackage{verbatim}
\usepackage{amsthm}
\usepackage{xcolor}
\usepackage{float}
\usepackage{amsmath}
\usepackage{subcaption}
\usepackage{graphicx}
\usepackage{dcolumn}
\usepackage{bm}
\usepackage[colorlinks=true, allcolors=blue]{hyperref}
\usepackage[mathlines]{lineno}

\newtheorem{corollary}{Corollary}
\newtheorem{definition}{Definition}

\begin{document}
\pagestyle{plain}

\title{Purification Strategy Optimization for Entanglement Routing in Quantum Networks}

\author{
\IEEEauthorblockN{Javier Vecino Peñas, Ana Fernández-Vilas, Rebeca P. Díaz-Redondo, Sergio Gándara Gándara, Manuel Fernández-Veiga}\\
\IEEEauthorblockA{
atlanTTic Research Centre, Universidade de Vigo
}
}

\maketitle

\begin{abstract}
Quantum networks rely on the efficient distribution of entanglement to enable long-distance quantum communication and information processing. A key challenge in these networks is the design of routing protocols capable of maintaining high-quality entanglement in the presence of noise, decoherence, and imperfect operations, which progressively degrade the fidelity of entangled states through entanglement swapping. Entanglement purification provides an effective mechanism to mitigate this degradation at the cost of additional resources. In this work, we study purification-aware quantum routing and formulate the problem of selecting optimal purification strategies as an optimization task. By employing dynamic programming techniques, we identify strategies that optimally balance resource consumption and end-to-end fidelity, demonstrating the effectiveness of our approach across different scenarios.
\end{abstract}

\section{INTRODUCTION}\label{intro}
Quantum networks (QNs) aim to extend quantum communication and information processing across geographically separated nodes by distributing entanglement as a fundamental resource. Such networks are expected to support a wide range of advanced functionalities, including secure quantum key distribution (QKD), distributed quantum computing, and network-assisted quantum sensing, thereby challenging conventional paradigms of classical communication systems.

In classical communication networks, the set of rules that determines how information is efficiently delivered across a network is referred to as a \textit{routing protocol}. In the quantum domain, the transmission and management of quantum information (qubits) are governed by quantum routing protocols. Several quantum routing algorithms have been proposed~\cite{Abane2025EntanglementRouting}, focusing on efficient information management while guaranteeing security through techniques inspired by quantum cryptography.

A key operational task in quantum routing protocols is the faithful transfer of quantum states across the network. This process is inevitably affected by environmental factors such as channel losses, decoherence, and imperfect node operations, which lead to a degradation of quantum states as they propagate through the network. To mitigate this degradation, one may resort to techniques based on local operations and measurements, commonly referred to as \textit{entanglement purification protocols}. Purification protocols consist of a sequence of operations applied to noisy quantum states in order to probabilistically increase their purity and fidelity despite the presence of noise.

Purification protocols can be applied iteratively to further counteract the effects of noise. The number of purification rounds applied to each entangled link in a quantum network can be represented as a tuple, which we refer to as a \textit{purification vector} or \textit{purification strategy}. When combined with a routing protocol, such a strategy gives rise to a purification-aware routing protocol, capable of not only delivering quantum information efficiently but also maintaining high levels of state fidelity across the network.

The objective of this work is to determine the purification strategy for a given quantum routing problem that maximizes the end-to-end output fidelity within a prescribed time threshold and under a set of operational constraints. Optimization problems of this nature are generally NP-hard, implying that exact solutions cannot be obtained in polynomial time for large network instances.

To address this challenge, we employ \textit{dynamic programming}. Dynamic programming is an algorithmic paradigm introduced by Richard Bellman in the 1950s~\cite{Bellman1957}, which simplifies complex optimization problems by decomposing them into a sequence of overlapping subproblems. By systematically precomputing and reusing solutions to these subproblems, dynamic programming enables a significant reduction in computational complexity and allows for the efficient exploration of optimal decision strategies.

This work is structured as follows. In Section~\ref{state}, we review previous works that provide the foundation for the present study. In Section~\ref{theory}, we introduce the theoretical framework required to describe the system and to characterize fidelity degradation mechanisms. Sections~\ref{model} and~\ref{solution} present the problem under consideration and the numerical methods employed to obtain its solution. Finally, in Section~\ref{results}, we discuss and interpret the results of the numerical simulations, and Section~\ref{conclusion} concludes the work with a summary and outlook.

\section{STATE OF THE ART}\label{state}
The theoretical and experimental foundations of quantum teleportation were established in the late 1990s, with some protocols demonstrating the transfer of unknown quantum states using entanglement and classical communication. Pioneering experiments by Bouwmeester et al., Pan et al. \cite{Pan2012Teleportation} and Bennett et al. \cite{Bennett1993Teleportation}, among others, showed the feasibility of teleporting quantum states over laboratory-scale distances and later over several kilometers of optical fiber and free-space links.

More recent research \cite{Huang_2022} explores entanglement purification in practical scenarios, including long-distance purification protocols using hyperentangled states and techniques suitable for integration into quantum repeater architectures. These works demonstrate measurable improvements in entanglement fidelity over noisy channels and highlight the utility of purification for real, noisy quantum networks.

In the context of network-level performance, studies \cite{campbell2024testinglinkfidelityquantum, Cicconetti2021} have addressed fidelity-guaranteed entanglement routing and resource scheduling for purification, illustrating how fidelity thresholds influence path selection, resource consumption, and throughput in multi-node networks. Such strategies are a key point when balancing the consumption of entangled pairs for purification against the need to maintain high end-to-end fidelity in complex network topologies.

Finally, it is important to remark that recent works have also addressed the optimization of entanglement distribution strategies in quantum networks. In particular, Huang \emph{et al.}~\cite{Huang2025DecoherenceAware} propose a decoherence-aware framework to optimize entangling and swapping strategies in entanglement routing employing novel techniques to deal with NP-hard optimization problems. Their approach involves the definition of auxiliary functions that do not explicitly depend on fidelity, thus simplifying the problem. This work focuses on optimizing the swapping and entanglement strategies, demonstrating that the skewer the swapping strategy is, the higher the end-to-end entanglement fidelity will be. In this work's context, a skew swapping strategy translates into swapping the initial nodes' entanglements first and then focusing on the rest of the entanglements. While this work focuses primarily on entanglement generation and swapping optimization, due to the simplifications made to avoid computing the overall fidelity, it ignores the role purification might play in real scenarios. Our work complements it by emphasizing the importance of entanglement purification and systematically optimizing purification strategies to enhance end-to-end fidelity.

\section{THEORETICAL BACKGROUND}\label{theory}

\subsection{Entanglement and Fidelity}
In section \ref{intro}, we stated that quantum networks rely on the entanglement between quantum states and that the purity of these entanglements decreases due to its interaction with the environment. To rigorously define what an entanglement is, let us first define the \textit{Bell states}. Suppose Alice \textit{A} possesses a qubit in superposition between $|0\rangle$ and $|1\rangle$, and suppose that Bob \textit{B} also possesses a superposed qubit between $|0\rangle$ and $|1\rangle$. If they were to measure their own qubit and Alice's measurements were correlated to Bob's, then their qubits are said to be \textit{entangled}. The entanglement between their qubits can be expressed as a linear combination of some vector basis of maximally entangled states, which are called \textit{Bell states}. This vector basis can be expressed as:

\begin{subequations}\label{eq:bells}
\begin{align}
        |\Phi^{\pm}\rangle &\equiv \dfrac{|0\rangle_A\otimes|0\rangle_B\pm |1\rangle_A \otimes |1\rangle_B}{\sqrt{2}},\label{bellphi}\\
        |\Psi^{\pm}\rangle &\equiv \dfrac{|0\rangle_A\otimes|1\rangle_B\pm |1\rangle_A \otimes |0\rangle_B}{\sqrt{2}}.\label{bellpsi}
\end{align}
\end{subequations}
In this context, a \textit{pure} entanglement will be described by the Bell state $|\Phi^+\rangle$ however, this state cannot remain \textit{pure} as time passes due to interactions with the environment, meaning that the correct way to express it is via a bipartite density matrix $\rho$ named \textit{Werner state} \cite{Werner1989, VanMeter2014} :
\begin{equation}\nonumber
\rho = \alpha\,|\Phi^+\rangle\langle\Phi^+| + \beta\,|\Phi^-\rangle\langle\Phi^-| +\delta \,|\Psi^+\rangle\langle\Psi^+| +\varepsilon\,|\Psi^-\rangle\langle\Psi^-|,
\end{equation} 
where $\alpha+\beta+\delta+\varepsilon =1$ and states $|\Phi^\pm\rangle$ and $|\Psi^\pm\rangle$ being the maximally entangled two-qubit Bell states described in equations \eqref{bellphi}-\eqref{bellpsi}. As the desired, pure state is the $|\Phi^+\rangle$ Bell state, we will rename its coefficient $\alpha$ to \textit{fidelity} $F$. The process in which the pure state $|\Phi^+\rangle$ interacts with the environment thus degrading its fidelity $F$ is called \textit{decoherence}. In general, an empirical formula for the process of decoherence exists \cite{PhysRevX.9.031045} and is described as follows:
\begin{equation}\label{decoherence}
    F_{d}(F,t) = F\cdot e^{ -\left(\gamma \cdot t\right)^\kappa},\quad t\ge0,
\end{equation}
where $\gamma$ is the decoherence rate measured in $\text{s}^{-1}$ and $t$ is the time measured in seconds. The variable $\kappa$ is a constant that encapsulates how strong the decoherence loss is, from now on will be set to one $\kappa = 1$. Since a quantum network relies on the entanglement between quantum states over long distances, a way of extending the entanglements is key to guarantee communication. Entanglement swapping is a technique used to \textit{lengthen} the entanglement distance by taking two entangled states that share a quantum node and performing a series of local measurements such that the resulting entanglements spans more distance than the individual initial entanglements. The fidelity of the output entanglement after swapping \cite{VanMeter2014} can be modeled as:
\begin{equation}\label{swapping}
    F_s(F_1, F_2) = F_1 \cdot F_2 + \dfrac{1}{3}(1-F_1)(1-F_2),
\end{equation} where $F_i$ denotes the fidelity of each input entanglements. The effect equation \eqref{swapping} has on the fidelities can be understood by the fact that the higher the input fidelities are, the less the fidelity will degrade after swapping two entanglements.

The process of purification is the process in which, by taking an auxiliary entangled pair and performing a sequence of actions to both states, the fidelity of the original entanglement is increased. Theoretically, a purification process can be repeated indefinitely but, in a real world scenario, the purification process would waste time and auxiliary entangled pairs. A plethora of purification protocols exists, but in this work we will focus on one of the most simple non-trivial one, named the IBM (BBPSSW) \cite{Bennett_1996} recursive protocol, whose recursive expression reads: 
\begin{equation}\label{Fid}\nonumber
    F'(F) =  \dfrac{10F^2-2F+1}{8F^2-4F+5}, \quad \text{where} \quad F\ge0.5.
\end{equation}
This recursive protocol increases the fidelity of an entanglement if and only if the input fidelity is above the value $0.5$ and stops purifying when fidelity reaches a value of $1.0$. In theory, every purification must be done before swapping in order to achieve higher fidelities overall. An important remark is that although we employ the $\text{BBPSSW}$ recursive protocol in this work, any kind of purification protocol can be used. This is because the main objective is increasing fidelities as much as we can, which any purification protocol achieves.

After describing all the processes that will come into play in the system (entanglements between nodes, entanglement swapping and purification processes) it is necessary to point out that, in a real system, each process has a distinct success probability. Starting with the generation of entanglements between two nodes, its success probability \cite{10.1145/3387514.3405853} decreases exponentially with distance:
\begin{equation}
    \text{P}_{ent} = 1- \left(1-e^{-\lambda \cdot L}\right)^\eta,\nonumber
\end{equation}
where $\lambda$ refers to an absorption coefficient, $L$ is the distance between two adjacent nodes and $\eta$ represents the number of entangling attempts. 
The success probabilities of swapping and purification will be labeled $\text{P}_{s}$ and $\text{P}_{pur}$ respectively, and in principle will be modeled as a number between 0 and 1. The combined success probability of a path $p$ \cite{10.1145/3387514.3405853} from the Source $(s)$ to the Destination $(d)$ will be:
\begin{equation}\nonumber
    \text{P}(p) = \prod_{v_i \in p} \text{P}_{ent}\cdot (\text{P}_{pur})^{n_{ij}} \cdot  \text{P}_s,
\end{equation}
where $\text{P}(p)$ denotes the success probability of the path, $p$ is the path from $s$ to $d$ and $n_{ij}$ is the number of purification rounds. 
\section{NETWORK MODEL AND PROBLEM FORMULATION} \label{model}
\subsection{Modeling the Quantum Network}\label{graph}
In order to model the Quantum Network, let us consider a graph $G = (V, E)$ where $V$ is the set of all the vertices, which will represent the number $|V|$ of quantum nodes of the network, and $E$ is the set of all the edges. Each quantum node $v_i$ possesses an amount of quantum memory $m$ representing the amount of qubits that are stored in it, up to a memory limit $m_{max}(v_i)$ for each node. Let us also define a SD pair (source-destination) $r\in R$ and a predefined path $p\in \mathcal{P}(r)$ calculated using one arbitrary path-finding algorithm. The total number of nodes that belong to the selected path $p$ can be defined as: $ N = |\{v_{i} :v_i \in p\}|. $
Time will be modeled as a discrete variable instead of a continuous one, and each discrete division will be referred to as \textit{time slot}. Each time slot will be denoted as $t\in T$ and $t_f(r)$ will be the final time slot used by every SD pair. The nodes of the network that are capable of establishing an entanglement between them weighed by the probability $\text{P}_{en}$ with an initial fidelity $F(v_i, v_j)$ are called \textit{adjacent} nodes. The entanglement that links two adjacent nodes $v_i$ and $v_j$ will be denoted as $\rho_{ij} = (v_i, v_j)$. The data qubits are transmitted via E2E entangled pairs in the QN. A long entangled pair may be generated by performing a swapping process between two entangled pairs that share a node, with a success probability $\text{P}_s$. Finally, the fidelity of an entangled pair can be purified with a probability $\text{P}_{pur}$ by employing an auxiliary entangled pair that is not stored in the nodes' memories. We will consider that each node can perform one action per time slot, this includes idling, purification, swapping and entanglement generation. Since all these parameters must be taken into account when modeling the network, we need to figure out a rigorous way to describe how to handle the resources in the QN: 
\begin{definition}
We define the strategy tree $\Upsilon$ as an activity-on-vertex tree
structure, describing an entangling, swapping and purification strategy between node $v_1$ and node $v_N$ on a chosen path $p=\{v_1, v_2,\dots, v_N\}$. There are three distinct parts of a strategy tree: the entangling part $\Upsilon_{ent}$, the swapping part $\Upsilon_{s}$ and the purification part $\Upsilon_{pur}$.
\end{definition}
The entangling and swapping parts of the decision tree $\Upsilon$ describe how to properly manage the entangling and swapping strategies respectively, while the purification part focuses on the purification strategy applied to a selected path $p$ of the network. The entangling and swapping parts of the decision tree $\Upsilon$ are further covered in \cite{Huang2025DecoherenceAware}, which means that a definition of the purification tree will be needed. The purification part of the tree $\Upsilon_{pur}$ consists on leaves that represent entangled pairs $(v_i,v_j)_{n_{ij}}$ where $n_{ij}$ denotes the rounds of purification applied to that pair. As we previously commentated, in \cite{Huang2025DecoherenceAware} it is shown that, the skewer the entangling and swapping trees are, the higher the final fidelity will be, which translates into a purification-only problem. 

To mathematically express how all these factors may influence the overall strategy and decision tree, let us define the following quantities:

\begin{definition}
A numerology $\mu$ represents the distribution of resources of an associated strategy tree. It can be formally defined as:
    \begin{equation}\nonumber
    \mu=\{(\rho_{ij}, S_{ij},T(\rho_{ij}, S_{ij}))\,|\,\rho_{ij}=(v_i,v_j)\in \Upsilon\},
\end{equation}
\end{definition}
\noindent which can be better understood by saying that a numerology is a tuple that encapsulates the time slots $T(\rho_{ij},S_{ij})$ \textit{used} by the pair $\rho_{ij}$ with $S_{ij}$ rounds of purification. In this context, $S$ is a vector whose components $S_{ij}$ indicate the number of purification rounds per entanglement. The set of all the possible numerologies of a problem will be denoted as $\mathcal{M}_p(r)$, where $p$ is the selected path and $r$ the SD pair.

\begin{definition}
The quantum memory of a quantum node $m_{\mu}$ is the cardinality of the associated numerology $\mu$ fulfilling some conditions. The memory can be mathematically described as:
    \begin{align}
    m_{\mu} (t,v_i) = |\{ (\rho_{ij}, S_{ij,}T(\rho_{ij},S_{ij})) \in \mu
     \nonumber\}| 
\end{align}
where $v_i\in\rho_{ij}$ and $t\in T(\rho_{ij},S_{ij})$.
\end{definition}

\begin{definition}
    A \textit{purification strategy} $S$ is a positive integer-valued vector whose components are defined as:
    \begin{equation}\nonumber
        S\equiv\left(n_{1,2},\, n_{2,3},\, \dots,\,n_{N-1,N},\, n_{1,3},\,n_{1,4},\,\dots, \, n_{1,N} \right),
    \end{equation}
    where $n_{i,j}$ is the number of purification rounds applied to entanglement $ij$ and $N$ is the number of nodes of the path we selected. 
\end{definition}

\begin{corollary}
    The dimension of the purification strategy $S$ is $2N-3$, where $N$ is the total number of nodes of the selected path.
\end{corollary}
\begin{proof}
    Let us split the strategy $S$ into 2 different vectors: symmetric-base vector $s$ and asymmetric vector $a$. We may write $S$ as follows: $S = (s,a) \equiv s \cup a$ where $s \equiv\left(n_{1,2},\,\, \dots,\,n_{N-1,N}\right)$ and $a \equiv \left( n_{1,3},\,\dots, \, n_{1,N} \right)$.

    The dimension of $s$ is just the number of links (entanglements) between $N$ nodes in a simple chain (path), which is always $N-1$.
    The dimension of $a$ can be calculated by counting its components. Suppose that we had all possible combinations of $n_{1,j}$ as components of a vector $v$. Its dimension would be: 
    \begin{equation}
        \text{dim}(v) = \sum_{j=1}^N n_{1,j} = N,\nonumber
    \end{equation}
    but, since $n_{1,1}$ is a forbidden component because a node cannot entangle with itself and $n_{1,2} \in s$, the vector $v$ becomes $a$. Recall that we have removed 2 components from vector $v$ and arrived at vector $a$, thus yielding:
    \begin{equation}
        \text{dim}(a) = \sum_{j=3}^N n_{1,j} \equiv N-2,\nonumber
    \end{equation}
    and now combining both results we arrive at:
    \begin{equation}
        \text{dim}(S) = \text{dim}(s\cup a) = \text{dim}(s) + \text{dim}(a) = 2N -3.\nonumber
    \end{equation}
\end{proof}

In order to visualize in a more clear way all the parameters we defined, let us consider the following example: a small path comprised of 4 quantum nodes inside an arbitrary quantum network. Suppose that the initial fidelities are the tuple: $ F = (F_{1,2}, F_{2,3}, F_{3,4})=(0.86,\,0.73,\,0.9)$ and we select 3 purification strategies (with no decoherence) to see the impact purification has on the E2E fidelity: 
\begin{enumerate}
    \item \textbf{No purification}: if we do not purify our states, then first applying equation \eqref{swapping} to the first components of $F$ yields a fidelity $F_{1,3}$ with a value of $0.64$, and after the final swapping the final Fidelity will be $F_{1,4}=0.588$. This case corresponds to the cases studied by \cite{Huang2025DecoherenceAware} and it is represented in figure \ref{fig:dtnp}.
    \item \textbf{Purification of the first pair}: this strategy corresponds to figure \ref{fig:dtp} where the red box indicates the purification attempts (represented by a zigzag line) and, thanks to purification, the initial two fidelities become $(F_{1,2}, F_{2,3}) =(0.89,\,0.77)$ respectively, after their swapping the fidelity drops to $F_{1,3} =0.69$ and in the final swapping we arrive at $F_{1,4}=0.63$.
    \item \textbf{One round of purification everywhere}: exactly the same as the previous case but with the addition of 2 more purifications before and after the final swapping. Recall that, before the final step, the fidelities are $F_{1,3}=0.69$ and $F_{1,4}=0.9$. After one round of purification we get $(F_{1,3}, F_{3,4})=(0.72,,0.93)$ and for the final swapping, fidelity drops to $F_{1,4}=0.67$.
\end{enumerate}

\begin{figure}[hbtp]
    \centering
    \includegraphics[width=0.9\linewidth]{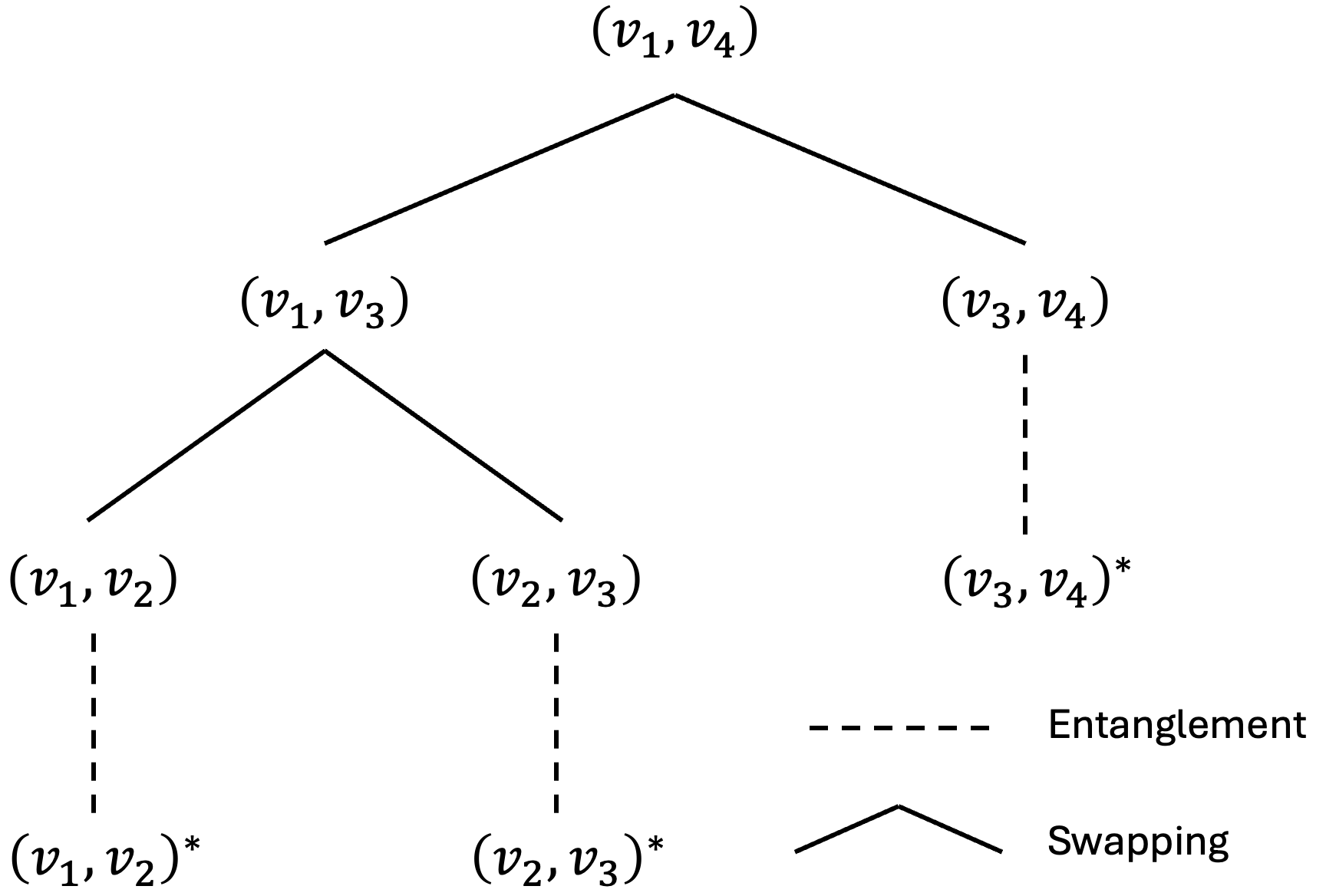}
    \caption{Decision tree for a 4-node path with no rounds of purification. States marked with an asterisk ($\ast$) refer to the creation of a pair. The idling of the nodes has been omitted.}
    \label{fig:dtnp}
\end{figure}

\begin{figure}[hbtp]
    \centering
    \includegraphics[width=1.0\linewidth]{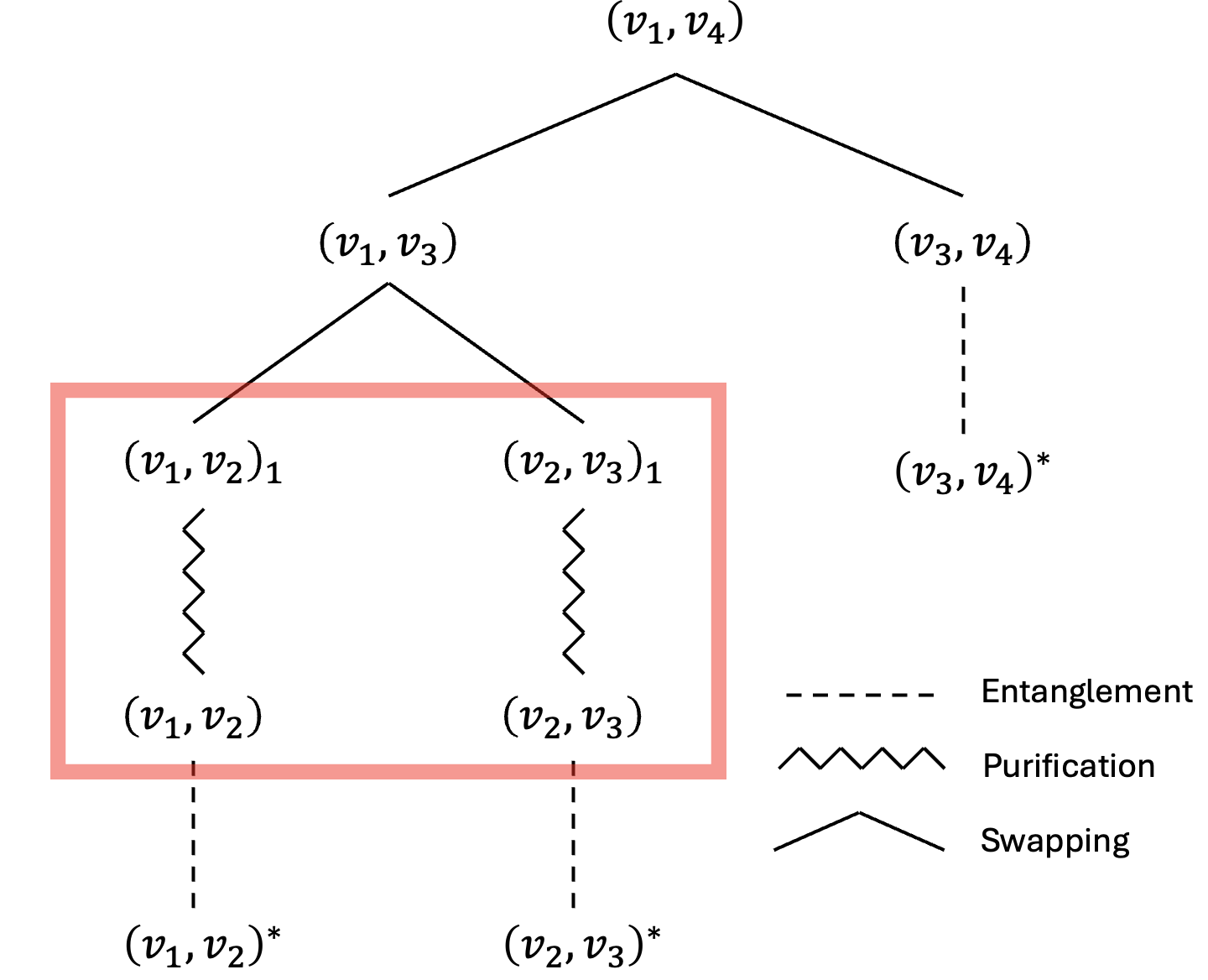}
    \caption{Decision tree for a 4-node path with one round of purification on the first pairs, represented by a red box. States marked with an asterisk ($\ast$) refer to the creation of a pair. The idling of the nodes has been omitted.}
    \label{fig:dtp}
\end{figure}

\subsection{Formulating the optimization problem} 
To mathematically define our problem, let us start by taking the QN represented by the graph $G$ explained on subsection \ref{graph}. Given a time threshold $\hat{T}$ and a fidelity threshold $\hat{F}$ the goal is to maximize the E2E fidelity $F(\mu)$ with the following constraints:

\begin{enumerate}
    \item The final time slot $t_f(r)$ for every SD pair shall not be greater than the threshold $\hat{T}\in T$.
    \item The final value of the Fidelity $F(\mu)$ must be greater than the threshold $\hat{F}$.
    \item The amount of memory used on any node $m_\mu(t,v_i)$ must be less than the maximum memory $m_{max}(v_i)$ at any time.
    \item Only one numerology $\mu\in \mathcal{M}_p(r)$ must be chosen.
\end{enumerate}

Before stating the complete problem, it is useful to define a tuple containing the SD pair $r$, the threshold time $\hat{T}$ and the threshold fidelity $\hat{F}$ which will be referred to as a \textit{request}, denoted by $\mathcal{R}$. The complete protocol should be able to efficiently manage this tuple $\mathcal{R}$ while employing purification processes to reach the threshold fidelity. Therefore, the complete problem can be expressed as:

\begin{subequations}\label{eq:opt_problem}
\begin{align}
    \text{Maximize}\quad
    & \sum_{\mathcal{R},\,p,\,\mu}
    F(\mu)\cdot\delta_{\mathcal{R}p\mu}, \label{eq:opt_obj} \\
    \text{Subject to}\quad
    & \sum_{\mathcal{R},\,p,\,\mu} m_\mu(t,v_i)\cdot\delta_{\mathcal{R}p\mu}
    \le m_{\max}(v_i), \label{eq:opt_c1} \\
    & \sum_{p,\,\mu} \delta_{\mathcal{R}p\mu} \le 1,
    \label{eq:opt_c2} \\
    & (t_f(r)-\hat{T})\cdot\delta_{\mathcal{R}p\mu} \le 0,
    \label{eq:opt_c3} \\
    & (F(\mu)-\hat{F})\cdot\delta_{\mathcal{R}p\mu} \ge 0,
    \label{eq:opt_c4} \\
    & \delta_{\mathcal{R}p\mu} \in \{0,1\},
    \label{eq:opt_c5} \\
    &\nonumber \\
    & \nonumber\forall r\in R,\;
      \forall p\in P(r),\;
      \forall \mu\in \mathcal{M}_p(r),\;\\
    & \nonumber\forall t\in T,\;
      \forall v_i\in V, \; \forall \mathcal{R},
\end{align}
\end{subequations}

where $\delta_{\mathcal{R}p\mu}$ equals 1 or 0 either a certain numerology $\mu$ is chosen or not. This corresponds to an integer maximization problem where the decision variable are integers. This kind of problems tend to be NP-hard problems that can be simplified by relaxing some of the conditions and constraints of the problem, such as promoting the integer-valued variables to real-valued ones or removing terms from the objective function. These simplifications translated into the problem \eqref{eq:opt_obj}-\eqref{eq:opt_c5} mean that the decision variable $\delta_{\mathcal{R}p\mu}$ would become a number between zero or one and that, in the objective function, the fidelities $F(\mu)$ would be downgraded into a constant or a fixed set of values. These exact procedure was adopted in \cite{Huang2025DecoherenceAware} using arguments about maximum and minimum fidelity on a QN, which would be impossible to implement alongside purification due to the possible values the E2E fidelity could take under purification protocols.

\section{FINDING SOLUTIONS TO THE PROBLEM}\label{solution}

\subsection{Brute-forcing a numerical solution}
Now that the problem has been formally defined, we should find a way to solve it. The first step is studying whether a solution can be found in polynomial time, exponential time, logarithmic time, etc. The main way to determine the complexity of a problem is to find a transformation that, in polynomial time, maps the main problem to another problem whose growth over time is already known \cite{HTCSVolA}. Since our goal is studying how purification affects fidelity, we cannot simplify the problem enough to be easily comparable to other well-known problems. The workaround to this problem is to check if, by brute-forcing all the possible solutions into the problem for a given number of nodes and parameters, we are able to find a strategy in a \textit{reasonable} amount of time. 

We developed a Python code to brute-force all possible strategies subject to the constraints stated in \eqref{eq:opt_c1}-\eqref{eq:opt_c5}: we assume that an ideal path between the source and destination has already been found and is composed of $N$ nodes. Then we feed the code with all the physical parameters (duration of a time slot, decoherence, purification methods, etc) and adopt a swapping strategy such that the first swaps occur between the nodes closest to the source until we reach the destination node. Then, we select a maximum number of purification rounds $n_{max}$, the number of iterations per execution, and let the program calculate the E2E fidelities and final time.

The output of these simulations contains the mean fidelity of each strategy and orders them from highest to lowest fidelities, thus ranking all the possible strategies. The next natural step is to vary any of the physical parameters of the network to see how the system reacts to fundamental changes. The success probabilities of swapping, entanglement, and purification do affect the system but in a trivial way: the closer they are to the value of 1.0, the less it affects the fidelity. The parameter that affects the system the most is the decoherence rate $\gamma$. We can see this effect figure \ref{fig:N7R320}, where we numerically simulated a path comprised of $N=7$ nodes with $n_{max} = 3$ maximum rounds of purification under 20 iterations.

\begin{figure}[hbtp]
    \centering
    \includegraphics[width=0.95\linewidth]{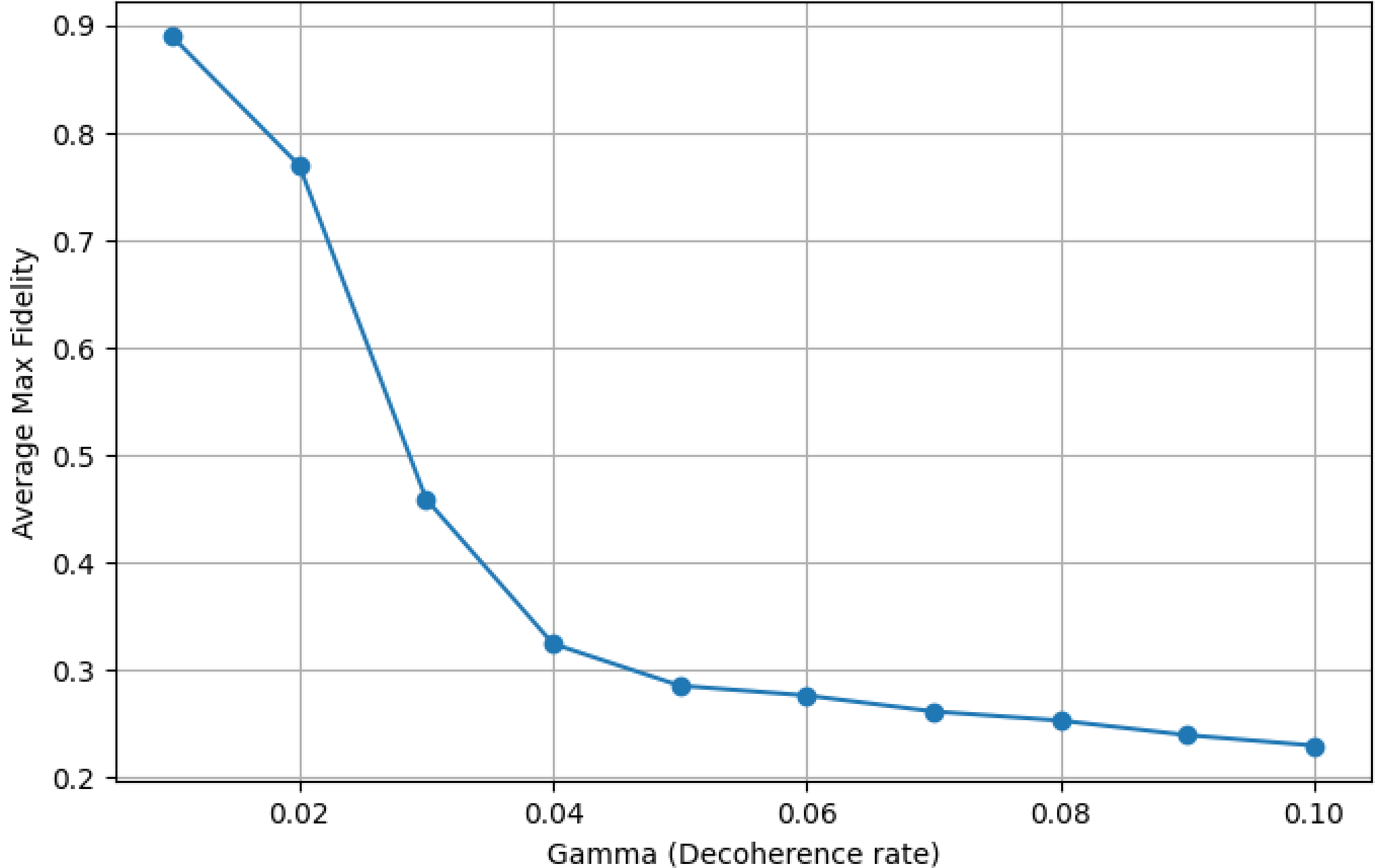}
    \caption{Average fidelity vs decoherence  $\gamma$ for 7 nodes, 3 maximum rounds of purification and 20 iterations.}
    \label{fig:N7R320}
\end{figure}

 The information we can extract from figure \ref{fig:N7R320} is that fidelity decreases as $\gamma$ increases, which means that in the fight between how much time is spent purifying versus the time the entanglement has to survive, the clear winner is decoherence due to the fact that BBPSSW cannot compete against the exponential degradation of decoherence \eqref{decoherence}. There are some ways of mitigating this fact, one being setting the maximum purification rounds to a low number, meaning that purification would take less time, or considering that thanks to the improvements in quantum memories, very low decoherence rates will be achieved.

In total, the key information we learned by brute-forcing the optimal strategy is that, for a high number of nodes $N$ in a path, the number of strategies we need to evaluate grows exponentially with the number of nodes $N$, more specifically, the exact number of strategies is $(n_{max}+1)^{\text{dim}(S)} = (n_{max}+1)^{2N-3}$, which means that our problem is likely to be NP-hard.

\subsection{Dynamic Programming Approach}
The exponential growth of the strategy space, defined by $(n_{max}+1)^{2N-3}$, makes brute-force optimization computationally intractable as the number of nodes $N$ increases. Since simplifying the model would require ignoring critical fidelity fluctuations, we employ Dynamic Programming (DP) to maintain physical accuracy while achieving computational efficiency. DP is an algorithmic paradigm that breaks down a complex optimization problem into a sequence of simpler subproblems. Then, we can pre-calculate and store optimal intermediate decisions, avoiding the redundant computations inherent in the recursive nature of quantum repeater chains.

To illustrate the necessity of this transition, consider the calculation of the $n$-th Fibonacci number ($F_n$). Naively, you would need to calculate the $n-1$-th and $n-2$-th numbers: $F_n = F_{n-1} + F_{n-2}$, where $F_n$ denotes the $n$-th Fibonacci number. The problem is that these numbers need $F_{n-2}$, $F_{n-3}$ and $F_{n-3}$ and $F_{n-4}$ respectively, which causes an exponential growth in time. To avoid this, we start by computing the first values of the sequence $F_0 = 0$, $F_1 = 1$ and then iteratively compute larger and larger numbers while storing these new values in order to calculate the next ones.

The technique of saving values that have already been calculated is called \textit{memoization} and although it sounds simple, it reduces the execution time significantly. In our case, we employ memoization to first compute the optimal strategy for a low number of nodes $n<N$ by brute-forcing the solution, then saving this strategy and using it as a starting point for the case of $n+1$ nodes. In this case, this approach is valid, as the addition of one node only creates two new entanglements without affecting the strategy or the entanglements of the previous case. This allows us to explore a reduced part of all the possible strategies' space, and, as in each step we just add the contribution of one new node, we are able to simulate the behavior of large systems without an exponential growth in execution time.

\subsection{Numerical Simulation of the selected path}
In order to simulate the system and compute optimal strategies, we developed another Python-written code that, given a starting and goal number of nodes, computes the best possible purification strategy $S$ using techniques from dynamic programming such as memoization. When generating the entanglements, we assign a random fidelity sampled from the range $[0.75, 0.99]$ \cite{Zhao2022E2E}. The duration of each time slot will be set to $1$ ms since all possible operations (purification, swapping, ...) take times of order $10^{-3}$ s \cite{Pompili2021NetworkTeleportation}. The probabilities of the operations are set to 1.0 due to the fact that the probabilities of the real scenarios oscillate between 0.9 and 0.95 \cite{10.1145/3387514.3405853, Pompili2021NetworkTeleportation, 9488850}. We adopted a skewed entangling and swapping strategy as it results in the highest possible E2E fidelity. Since optimizing both the E2E fidelity and the final time $t_f$ at the same time is impossible, we will focus only on maximizing the output fidelity. The other relevant physical parameter is the decoherence rate $\gamma$, which will be sampled between $\gamma \in [0,\infty)$. We ran tests from $N_{in} = 3$ to $N_{goal} = 200$, which means that visualizing the computed strategies is difficult due to the large amount (recall: $\text{dim}(S) = 2N-3$) of the strategy's components. For example, let us see what happens to the first three strategies in a $\gamma = 2.0 \,\text{Hz}$ run with 8 maximum rounds of purification: 
\begin{eqnarray}
    N=3 \Rightarrow S&=&\left(8,8,8\right),\nonumber\\
    N=4 \Rightarrow  S&=&\left(8,8,1,8,1\right),\nonumber\\
    N=5 \Rightarrow S&=&\left(8,8,1,1,8,1,1\right)\nonumber,
\end{eqnarray}
notice that, as we implemented memoization, the values that were previously calculated stay on the next iteration, thus forming pyramid-like structures. Visualizing the evolution of the strategy with the number of nodes is a hard task, and if we want to visualize the strategies for large number of nodes in a more comprehensible way, we will need to adjust \textbf{the order of the components of the strategy}, adding the new components computed on each round in the \textit{rightmost} part of the vector. For instance, if we focus on the previous case: 
\begin{eqnarray}
     S_3&\equiv &\left(8,8,8\right),\nonumber\\
   S_4 &\equiv & \left(8,8,8,1,1\right) = (S_3,1,1),\nonumber\\
   S_5&\equiv &\left(8,8,8,1,1,1,1\right) = (S_4, 1, 1)\nonumber,
\end{eqnarray}
where we define $S_N$ as the strategy for the $N$-node case with shifted components. Recall that we are changing the order of the components to visualize the complete purification strategy vector. Now, for the rest of the strategies in the case $\gamma = 2.0 \, \text{Hz}$, let us focus on some significant values:

\begin{alignat}{2}
    & S_{10} = (S_3, \vec{1}_{14}),  & \hspace{1cm} & S_{15} = (S_3, \vec{1}_{24})\nonumber \\
    & S_{26} = (S_3, \vec{1}_{45}, 0), & & S_{50} = (S_3, \vec{1}_{45}, \vec{0}_{49}) \nonumber\\
    & S_{100} = (S_3, \vec{1}_{45}, \vec{0}_{149}), & & S_{200} = (S_3, \vec{1}_{45}, \vec{0}_{349})\nonumber
\end{alignat}

where $\vec{0}_{n}\, \,\text{and}\,\,\vec{1}_n$ stands for a vector of dimension $n$ whose components are all zeros or ones, respectively. Although rigorous, these representations of the results are not visual enough, which means that turning to a graphical way of representing the results may help us more when interpreting the results. One way to graph the behavior of the strategies is to compute the mean number of purification rounds for each value of $N$ and then plot it against the number of nodes.
\begin{figure}[hbtp]
    \centering
    \includegraphics[width=0.95\linewidth]{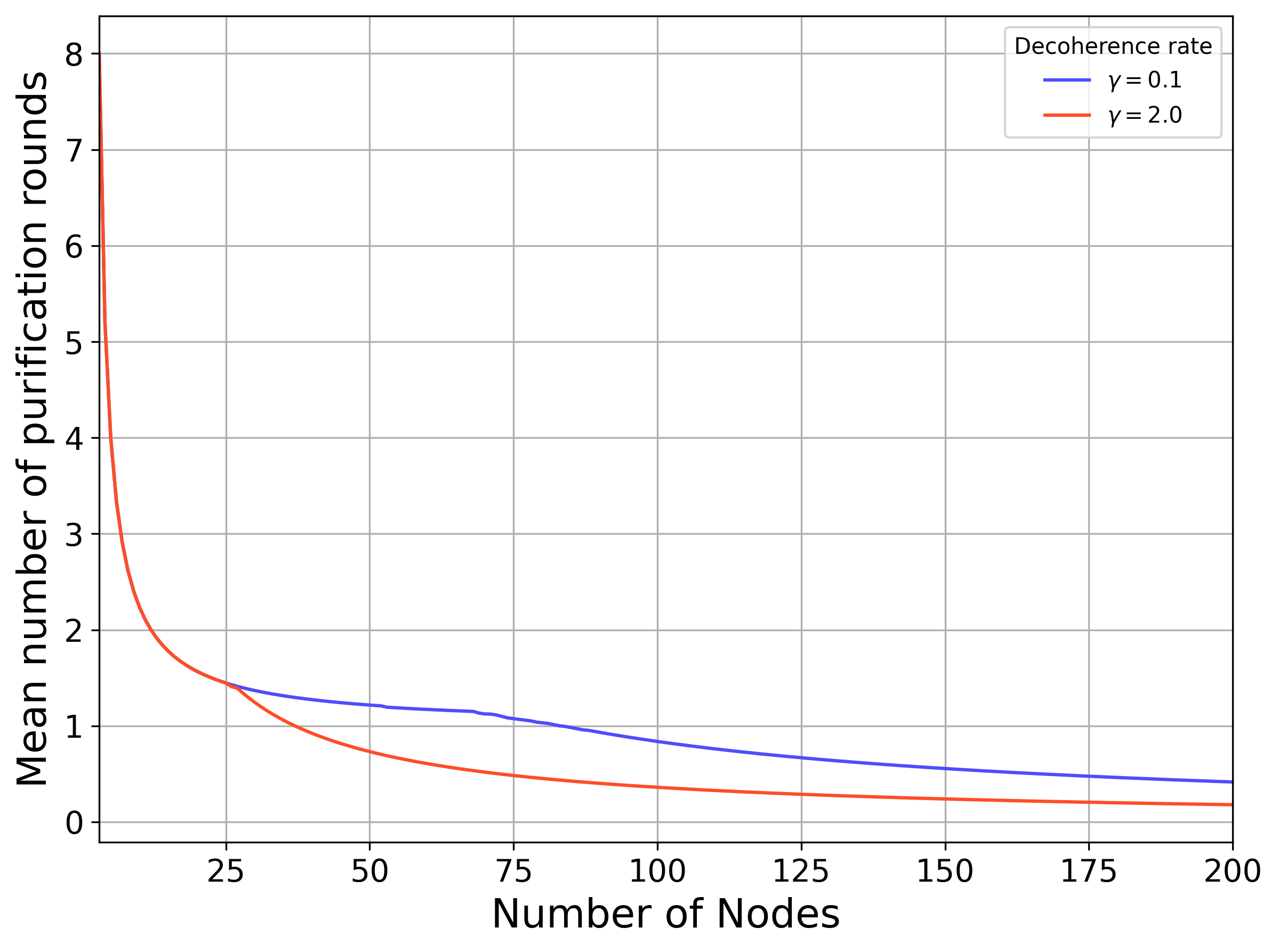}
    \caption{Mean purification rounds against number of nodes N, decoherence rate $\gamma = 0.1,\, 2.0 \, \text{Hz}$.}
    \label{fig:0.1}
\end{figure}
\begin{figure}
    \centering
    \includegraphics[width=0.95\linewidth]{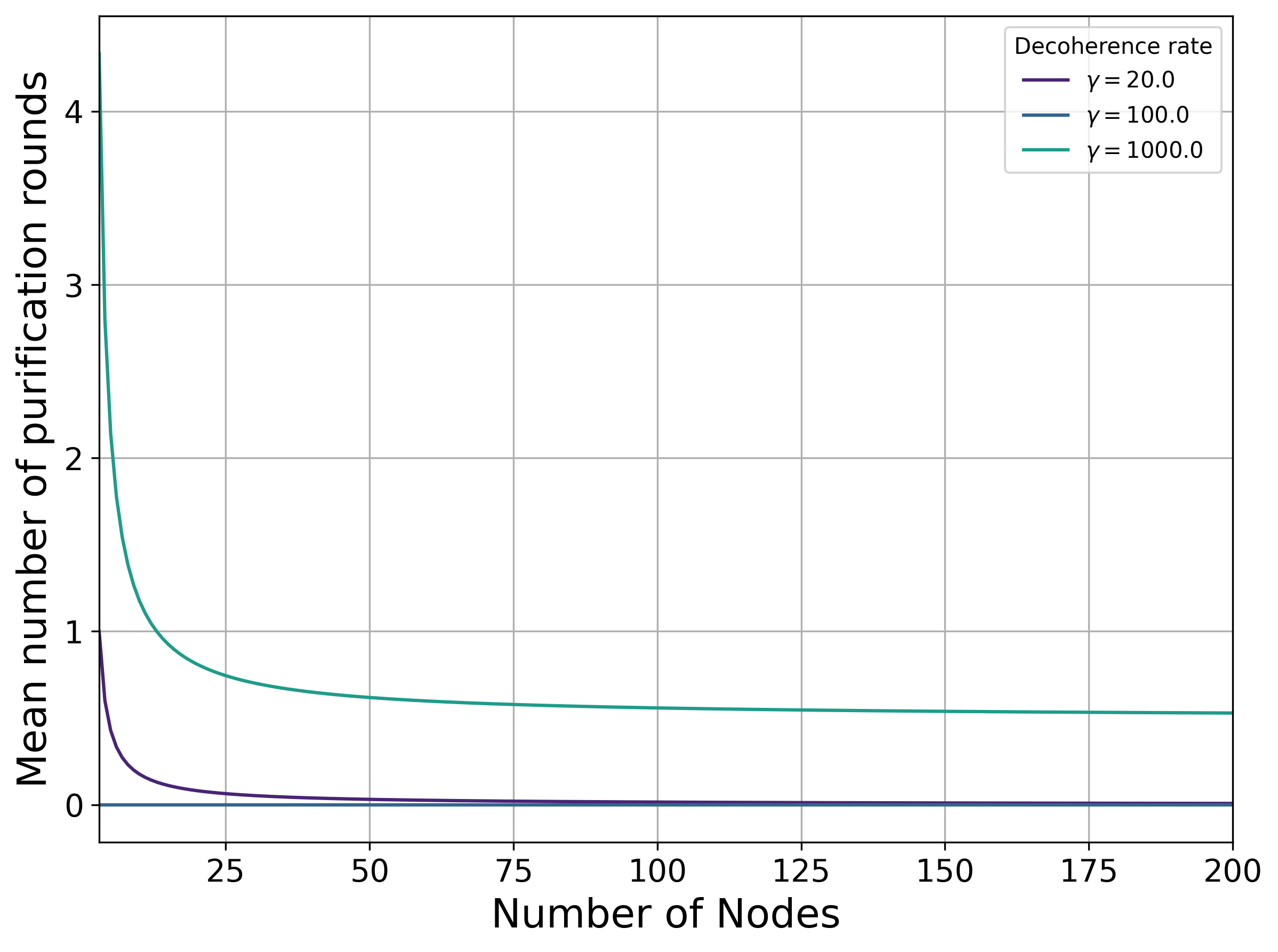}
    \caption{Mean purification rounds against number of nodes N, decoherence rate $\gamma = 20.0, \, 100.0,\, 1000.0 \,\text{Hz}$.} 
    \label{fig:1000.0}
\end{figure}

By doing so, we obtain figures~\ref{fig:0.1}-\ref{fig:1000.0} which seem to show that, for a low number of nodes (first rounds that can be purified in large paths), the best strategy is to purify as much as we can since the mean number of purification rounds is high as $N \rightarrow 0$. As $N$ grows, it is more efficient to purify less and less rounds.

If we look at Fig. \ref{fig:1000.0}, the cases where $\gamma$ starts to increase seem counterintuitive. The cases where $\gamma$ equals $20.0$ and $100.0$ Hz show that, for increasing values of the decoherence, purification stops being useful: entanglements lose fidelity faster than the gains from the purification processes. We may expect the case where $\gamma = 1000.0$ Hz to behave in the same way; however, we can predict that when the product $\gamma \cdot t$ tends to one, the system must react and behave differently. If we recall equation \eqref{decoherence} then it is clear that, in the $\gamma \cdot t \sim 1$ regime, fidelity will drop by a factor $1/e$ at each time slot, triggering a \textit{phase transition} in the system. A phase transition \cite{FernandezPinedaVelasco2009} is a process in which a physical system starts to behave in a different manner, resulting in an overall change in its fundamental defining properties. Examples of this phenomenon can be seen on the different states of matter, the Ising model for magnetic materials, or super-fluidity. This effect can be seen in figure \ref{fig:transition} where, for values of $\gamma$ close to $1000.0$ Hz, the purification strategies do not differ significantly from one another. The presence of quantities or parameters that seem to \textit{freeze} when varying the conditions of a physical system is a clear indicator of phase transitions. To visualize this \textit{freezing}, let us imagine the transition from water to water vapor: water boils at $100 \,\,^{\text{o}}\text{C}$ under normal conditions and stays at this temperature until all of the water becomes water vapor. In this case, the strategy vector's components play a similar role: they do not change until the system \textit{crosses} the point where $\gamma \cdot t \sim 1$.

\begin{figure}[hbtp]
    \centering
    \includegraphics[width=1.0\linewidth]{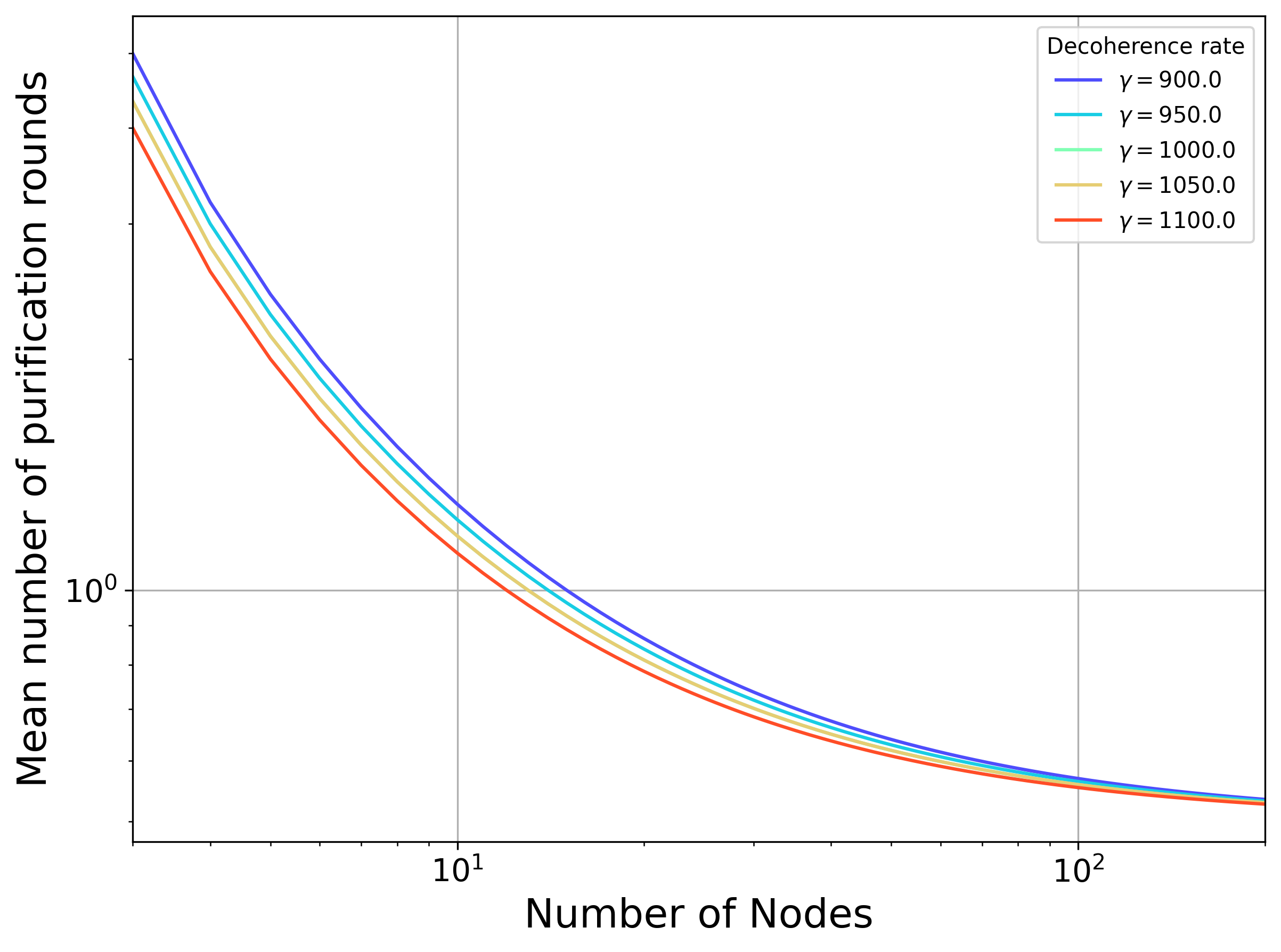}
    \caption{Mean number of purification rounds against the number of nodes $N$, decoherence rates near the transition, log scale. The $\gamma = 1000.0$ and $\gamma = 1050.0$ lines overlap one another.}
    \label{fig:transition}
\end{figure}

\section{EVALUATION OF THE RESULTS}\label{results}
After adopting a memoization technique and running some tests, the results seem to show that if you want to improve the E2E fidelity then the entanglements must be purified as early in the entanglements' life as you can and as many rounds as the decoherence rate allows you. This results do agree with the theoretical predictions about the following concepts:

\begin{itemize}
    \item Effect of \textbf{decoherence}: as the longer an entanglement \textit{lives} the lower the fidelity will get, if we want these entanglements to be purified, the purification must be done when the least number of entanglements are present. Since we adopted a skewed swapping scheme, this translates into purifying at the first steps of each entanglement.
    \item Fidelity degradation due to \textbf{swapping}: in general, the fidelity after performing entanglement swapping is much less than the fidelities before swapping. It makes sense to purify before any entanglement swapping, which is a previously known result. If we wanted to counteract this degradation, we would have to perform a high number of purification rounds, which would activate the decoherence thus degrading fidelity even more.
\end{itemize}

The effect of decoherence is the most impactful parameter of the simulation and, depending on its value, the optimal strategy may shift to a more \textit{passive} style of purification (less purification rounds) to a more \textit{aggressive} strategy (higher purification rounds). What we can observe is that, for low decoherence rates $\gamma \sim 1$, the best strategy seems to be purify the initial entanglements all the rounds the decoherence allows you until the fidelity losses outweigh the purification protocols. In Figure~\ref{fig:0.1} we can clearly see the \textit{bumps} where the loss due to decoherence overcomes the effect of purifying the entanglements.

When dealing with high decoherence rates, the best strategy is not purifying at all and letting the system evolve freely. As we discussed earlier, the moment $\gamma$ starts increasing, the system desperately needs a way to reach the fidelity threshold, causing nonsensical purification strategies.

It is clear that different strategies may achieve similar or even equal E2E fidelities, which implies the non-existence of an absolute \textit{best strategy} in each case; and that, depending on the decoherence rates of your system, sometimes is better to purify at the beginning or not purifying at all.

\section{CONCLUSIONS AND FUTURE WORKS}\label{conclusion}

In this work, we investigate the role that entanglement purification processes play in enhancing the fidelity of end-to-end quantum communications. Our study consisted in understanding how different purification strategies impacted the performance of multi-node quantum networks. 

By applying techniques borrowed from dynamic programming, we were able to compute optimal purification strategies for different network scenarios. These strategies consistently maximize the E2E fidelity while managing the losses caused by decoherence. Our results seem to indicate that the earlier the purifications take place, the higher the E2E fidelity will be.

The results we obtained offer a guide on how to purify any general path between the source and destination in a quantum network. However, in a real-world quantum network, the number of nodes between the source and destination nodes would be small, and advances in technologies regarding quantum memories \cite{Mamann2025, kbwj-md9n} translate into decoherence rates that are more forgiving, which means that the optimal strategy is to purify as early as possible.

Although these results provide some understanding of the behavior of systems under purification protocols, it is important to mention that the processes that take place in real life are subject to more complex environmental interactions. Future work could extend this approach to more realistic interactions and probabilities, include path-finding algorithms within complex quantum networks, dive deeper into the phase transition phenomenon, or explore additional classes of purification protocols to further improve practical applicability. 

Overall, our findings demonstrate that proper planification of a purification strategy under the constraints ruled by the different sources of fidelity loss is the key to achieve routing protocols that guarantee high E2E fidelities.

\section*{Acknowledgments}
This work was supported in part by MCIU/AEI/10.13039/501100011033 (DISCOVERY Project) under Grant PID2023-148716OB-C31; in part by the ‘‘TRUFFLES: TRUsted Framework for Federated LEarning Systems, within the Strategic Cybersecurity Projects (INCIBE, Spain), funded by the Recovery, Transformation and Resilience Plan (European Union, Next Generation)’;’  in part by the Galician Regional Government under Project ED431B 2024/41 (GPC); and in part by Grant RETECH-LabCiber5G/6G funded by INCIBE with support from the “European Union NextGenerationEU/PRTR” and co-funded by Universidade de Vigo.

This research project was made possible through the access granted by the Galician Supercomputing Center (CESGA) to its supercomputing infrastructure. The supercomputer FinisTerrae III and its permanent data storage system have been funded by the NextGeneration EU 2021 Recovery, Transformation and Resilience Plan, ICT2021-006904, and also from the Pluriregional Operational Programme of Spain 2014-2020 of the European Regional Development Fund (ERDF), ICTS-2019-02-CESGA-3, and from the State Programme for the Promotion of Scientific and Technical Research of Excellence of the State Plan for Scientific and Technical Research and Innovation 2013-2016 State subprogramme for scientific and technical infrastructures and equipment of ERDF, CESG15-DE-3114


\bibliographystyle{unsrt}   
\bibliography{Biblio}

\end{document}